\documentclass[twocolumn,conference]{IEEEtran}
\usepackage[T1]{fontenc}
\usepackage[latin9]{luainputenc}
\usepackage{textcomp}
\usepackage{amsmath}
\usepackage{amsthm}
\usepackage{amssymb}
\usepackage{stackrel}
\usepackage{graphicx}
\usepackage[unicode=true,
 bookmarks=true,bookmarksnumbered=true,bookmarksopen=true,bookmarksopenlevel=1,
 breaklinks=false,pdfborder={0 0 0},pdfborderstyle={},backref=false,colorlinks=false]
 {hyperref}
\hypersetup{pdftitle={Your Title},
 pdfauthor={Your Name},
 pdfpagelayout=OneColumn, pdfnewwindow=true, pdfstartview=XYZ, plainpages=false}

\makeatletter

\newcommand{\lyxmathsym}[1]{\ifmmode\begingroup\def\b@ld{bold}
  \text{\ifx\math@version\b@ld\bfseries\fi#1}\endgroup\else#1\fi}

\theoremstyle{plain}
	  \newtheorem{prop}{Proposition}
\ifx\proof\undefined\
  \newenvironment{proof}[1][\proofname]{\par
    \normalfont\topsep6\p@\@plus6\p@\relax
    \trivlist
    \itemindent\parindent
    \item[\hskip\labelsep
          \scshape
      #1]\ignorespaces
  }{%
    \endtrivlist\@endpefalse
  }
  \providecommand{\proofname}{Proof}
\fi

 \usepackage{amsmath} 
\usepackage[numbers,sort&compress]{natbib}
\renewcommand{\maketag@@@}[1]{\hbox{\m@th\normalsize\normalfont#1}}
\usepackage{
    graphicx}  
\usepackage{
    float} 
\usepackage{
    subfigure}  
\usepackage{etoolbox,xstring,mfirstuc,textcase}
\usepackage[acronym]{glossaries}
\newcommand{\newac}{\newacronym}
\newcommand{\ac}{\gls}
\newcommand{\Ac}{\Gls}
\newcommand{\acpl}{\glspl}

\makeglossaries
\newac{isac}{ISAC}{integrated sensing and communication}
\newac{an}{AN}{artificial noise}
\newac{csi}{CSI}{channel state information}
\newac{fp}{FP}{fractional planning}
\newac{bme}{BME}{beampattern matching error}
\newac{ula}{ULA}{uniform linear array}
\newac{cu}{CU}{communication user}
\newac{uavs}{UAVs}{unmanned aerial vehicles}
\newac{sinr}{SINR}{signal-to-interference-plus-noise ratio}
\newac{sdr}{SDR}{semi-definite relaxation}
\newac{1d}{1D}{one-dimensional}
\newac{zf}{ZF}{zero-forcing}
\newac{aod}{AoD}{angle of departure}
\newac{los}{LoS}{line of sight} 
\newac{bs}{BS}{base station}
\newac{bss}{BSs}{base stations}
\newac{cus}{CUs}{communication users}
\newac{snr}{SNR}{signal-to-noise ratio}
\newac{evd}{EVD}{eigenvalue decomposition}
\newac{qsdp}{QSDP}{quadratic semidefinite programing}
\newac{misome}{MISOME}{MISO multiple eavesdroppers}

\makeatother

\begin{document}
\title{Optimal Transmit Beamforming for Secrecy Integrated Sensing and Communication}
\author{\IEEEauthorblockN{Zixiang Ren$^{1,2}$, Ling Qiu$^{1}$, and Jie Xu$^{2}$}\IEEEauthorblockA{$^{1}$Key Laboratory of Wireless-Optical Communications, Chinese
Academy of Sciences,}\IEEEauthorblockA{School of Information Science and Technology, University of Science
and Technology of China}\IEEEauthorblockA{$^{2}$School of Science and Engineering and Future Network of Intelligence
Institute (FNii),} The Chinese University of Hong Kong, Shenzhen \IEEEauthorblockA{E-mail: rzx66@mail.ustc.edu.cn, lqiu@ustc.edu.cn, xujie@cuhk.edu.cn}}
\maketitle
\begin{abstract}
This paper studies a secrecy integrated sensing and communication
(ISAC) system, in which a multi-antenna base station (BS) aims to
send confidential messages to a single-antenna communication user
(CU), and at the same time sense several targets that may be suspicious
eavesdroppers. To ensure the sensing quality while preventing the
eavesdropping, we consider that the BS sends dedicated sensing signals
(in addition to confidential information signals) that play a dual
role of artificial noise (AN) for confusing the eavesdropping targets.
Under this setup, we jointly optimize the transmit information and
sensing beamforming at the BS, to minimize the matching error between
the transmit beampattern and a desired beampattern for sensing, subject
to the minimum secrecy rate requirement at the CU and the transmit
power constraint at the BS. Although the formulated problem is non-convex,
we propose an algorithm to obtain the globally optimal solution by
using the semidefinite relaxation (SDR) together with a one-dimensional
(1D) search. Next, to avoid the high complexity induced by the 1D
search, we also present two sub-optimal solutions based on zero-forcing
and separate beamforming designs, respectively. Numerical results
show that the proposed designs properly adjust the information and
sensing beams to balance the tradeoffs among communicating with CU,
sensing targets, and confusing eavesdroppers, thus achieving desirable
transmit beampattern for sensing while ensuring the CU's secrecy rate.
\end{abstract}


\section{{Introduction}}

\Ac{isac} has been recognized as one of the potential key technologies
towards future B5G and 6G wireless networks \cite{LiuMasJ20,RahLusJ20,liu2021integrated},
in which sensing is integrated as a new functionality to enable emerging
environment-aware applications such as auto-driving, industrial automation,
and \ac{uavs} \cite{wUxUj21}. As compared to conventional wireless
networks with communication only,  \ac{isac} enables the dual use
of wireless infrastructures and scarce spectrum/power resources for
both communication and sensing, thus leading to enhanced system performance
at reduced cost.

The emergence of \ac{isac} introduces new data security issue for
wireless networks, especially when information-bearing signals are
reused for the purpose of sensing \cite{LiuMasJ18,LiuHuangNirJ20,hua2021optimal}.
For example, to ensure the performance of both sensing and communication,
the ISAC transmitters (e.g., cellular \acpl{bs}) need to focus  their
transmission (e.g., via transmit beamforming optimization \cite{LiuMasJ18,LiuHuangNirJ20,hua2021optimal})
towards both \acpl{cu} and targets. This, however, results in a severe
information leakage issue, as the sensing targets can be untrusted
and may eavesdrop the broadcast information signals \cite{wei2021multifunctional}. 

To deal with this issue, the physical layer security (see, e.g., \cite{GopPraJ08,WuKhiJ18})
has emerged as a viable solution to achieve perfectly secure information
transmission in \ac{isac}, by exploiting the difference between the
legitimate communication channel and wiretap channels of eavesdropping
targets \cite{WuKhiJ18}. In physical layer security, the secrecy
rate is normally adopted as the performance metric, which is defined
as the maximum transmission rate at which eavesdroppers are unable
to decode any information \cite{GopPraJ08}. Notice that the physical
layer security is much easier to be implemented in \ac{isac} systems
than that in conventional secrecy communication systems. This is due
to the fact that the \ac{isac} transmitter can exploit the integrated
sensing functionality to obtain the location and \ac{csi} of eavesdropping
targets, thus resolving the eavesdroppers' \ac{csi} issue faced in
conventional secrecy communication systems. 

To enhance the secrecy rate in \ac{isac} systems while ensuring the
sensing requirements, \ac{an} has been widely employed to not only
confuse the untrusted targets' eavesdropping, but also play a dual
role of sensing these targets \cite{wei2021multifunctional}. For
instance, the authors in \cite{DelAnaJ18} studied the secrecy \ac{isac}
system with one \ac{cu} and one eavesdropping target, in which the
transmit covariance matrices of both information signals and \ac{an}
were jointly optimized to maximize the \ac{cu}'s secrecy rate while
ensuring the received \ac{sinr} for sensing. The authors in \cite{SuLiuChrC19,SuLiuChrJ21}
further considered the scenario with multiple \acpl{cu} and one eavesdropping
target, in which the \ac{snr} at the eavesdropping target was minimized
while ensuring the \ac{cu}s' individual \ac{sinr} constraints and
certain sensing beampattern requirements. Due to the non-convexity
of the formulated problems, only sub-optimal beamforming/precoding
designs were obtained in these prior works \cite{SuLiuChrC19,SuLiuChrJ21,DelAnaJ18}
under their respective setups. 

Different from prior works, this paper studies a new secrecy \ac{isac}
scenario with one multi-antenna \ac{bs}, one single-antenna \ac{cu},
and multiple sensing targets, in which a portion of targets are untrusted
suspicious eavesdroppers. Under this setup, we adopt the matching
error between the exactly achieved transmit beampattern and a desired
beampattern (with energy towards sensing targets) as the performance
metric for sensing, and directly use the secrecy rate as the performance
metric for secrecy communication. Our objective is to minimize the
beampattern matching error for sensing by jointly optimizing the transmit
information and sensing beamforming, subject to the minimum secrecy
rate requirement at the \ac{cu} and the transmit power constraint
at the BS. Although the formulated problem is non-convex in general,
we obtain its globally optimal solution by using the technique of
\ac{sdr} together with a \ac{1d} search, in which the tightness
of \ac{sdr} is rigorously proved. Next, to avoid the high complexity
induced by the \ac{1d} search, we further develop two sub-optimal
solutions based on \ac{zf} and separate beamforming designs, respectively.
Finally, numerical results show that the proposed designs properly
adjust the information and sensing beams to balance various tradeoffs
among communicating with \ac{cu}, sensing targets, and confusing
eavesdroppers, thus achieving desirable transmit beampattern for sensing
while ensuring the \ac{cu}'s secrecy rate. 

\textit{Notations: }Vectors and matrices are denoted by bold lower-
and upper-case letters, respectively. $\boldsymbol{I}$ and $\boldsymbol{0}$
represent an identity matrix and an all-zero vector/matrix with appropriate
dimensions, respectively. For a square matrix $\boldsymbol{A}$, $\textrm{Tr}(\boldsymbol{A})$
denotes its trace, and $\boldsymbol{A}\succeq\boldsymbol{0}$ means
that $\boldsymbol{A}$ is positive semi-definite. For an arbitrary-size
matrix $\boldsymbol{B}$, $\textrm{rank}(\boldsymbol{B})$, $\boldsymbol{B}^{T}$,
and $\boldsymbol{B}^{H}$ denote its rank, transpose, and conjugate
transpose, respectively. $\mathbb{E}(\cdot)$ denotes the stochastic
expectation, and $\|\cdot\|$ denotes the Euclidean norm of a vector.
$\mathcal{CN}(\boldsymbol{x},\boldsymbol{Y})$ denotes a circularly
symmetric complex Gaussian (CSCG) random vector with mean vector $\boldsymbol{x}$
and covariance matrix $\boldsymbol{Y}$. $(x)^{+}\triangleq\max(x,0)$.

\section{System Model }

We consider a secrecy ISAC system, which consists of a BS equipped
with a \ac{ula} with $N>1$ antenna elements, a \ac{cu} with one
single antenna, and $K$ sensing targets. Let $\mathcal{K}=\{1,...,K\}$
denote the set of targets, among which the first $K_{E}$ ones (with
$K_{E}\le K$) are assumed to be untrusted or suspicious eavesdroppers,
denoted by set $\mathcal{K}_{E}=\{1,...,K_{E}\}\subseteq\mathcal{K}$.
We consider that the BS uses the linear transmit beamforming to send
the confidential message $s_{0}$ to the \ac{cu}, where $s_{0}$
is a CSCG random variable with zero mean and unit variance, i.e.,
$s_{0}\sim\mathcal{CN}(0,1)$, and $\boldsymbol{w}_{0}\in\mathbb{C}^{N\times1}$
denotes the corresponding transmit information beamforming vector.
Besides the information signal $s_{0}$, the BS also sends dedicated
sensing signal or equivalently \ac{an} $\boldsymbol{s}_{1}\in\mathbb{C}^{N\times1}$
to facilitate target sensing and confuse eavesdropping targets at
the same time. Here, $\boldsymbol{s}_{1}$ is an independent CSCG
random vector with zero mean and covariance matrix $\boldsymbol{S}=\mathbb{E}(\boldsymbol{s}_{1}\boldsymbol{s}_{1}^{H})\succeq\boldsymbol{0},\textrm{i.e}.,$
$\boldsymbol{s}_{1}\!\sim\!\mathcal{CN}(\boldsymbol{0},\boldsymbol{S})$.
Note that $\boldsymbol{S}$ is assumed to be of general rank, i.e.,
$0\le m=\textrm{rank}(\boldsymbol{S})\le N$. This corresponds to
the case with $m$ sensing beams, each of which can be obtained via
the \ac{evd} of $\boldsymbol{S}$. As a result, the transmitted signal
by the \ac{bs} is expressed as\setlength{\abovedisplayskip}{3pt} \setlength{\belowdisplayskip}{3pt}
\begin{equation}
\boldsymbol{x}=\boldsymbol{w}_{0}s_{0}+\boldsymbol{s}_{1}.
\end{equation}
Suppose that the \ac{bs} is subject to a total transmit power budget
$Q$. We thus have\setlength{\abovedisplayskip}{3pt} \setlength{\belowdisplayskip}{3pt}
\begin{equation}
\mathbb{E}(\|\boldsymbol{x}\|^{2})=\textrm{Tr}(\boldsymbol{S})+\|\boldsymbol{w}_{0}\|^{2}=Q.\label{eq:power constraint}
\end{equation}

We consider a quasi-static channel model, in which the wireless channels
remain unchanged over the time block of our interest, but may change
from one block to another. Let $\boldsymbol{g}\in\mathbb{C}^{N\times1}$
denote the channel vector from the BS to the  \ac{cu}. Accordingly,
the received signal at the \ac{cu} is expressed as\setlength{\abovedisplayskip}{3pt} \setlength{\belowdisplayskip}{3pt}
\begin{equation}
y=\boldsymbol{g}^{H}\boldsymbol{w}_{0}s_{0}+\boldsymbol{g}^{H}\boldsymbol{s}_{1}+z,
\end{equation}
where $z\sim\mathcal{CN}(0,\sigma_{0}^{2})$ denotes the additive
white Gaussian noise (AWGN) at the \ac{cu} receiver. The received
\ac{sinr} at the \ac{cu} is\setlength{\abovedisplayskip}{3pt} \setlength{\belowdisplayskip}{3pt}
\begin{equation}
\gamma_{0}(\boldsymbol{w}_{0},\boldsymbol{S})=\frac{\big|\boldsymbol{g}^{H}\boldsymbol{w}_{0}\big|^{2}}{\boldsymbol{g}^{H}\boldsymbol{S}\boldsymbol{g}+\sigma_{0}^{2}}.
\end{equation}

We consider the \ac{los} channel from the \ac{bs} to both trusted
and untrusted targets, similarly as in prior works \cite{SuLiuChrC19,SuLiuChrJ21}.
Let $\theta_{k}$ denote the \ac{aod} from the BS to target $k$.
Accordingly, the steering vector with angle $\theta_{k}$ is given
as\setlength{\abovedisplayskip}{3pt} \setlength{\belowdisplayskip}{3pt}
\begin{equation}
\begin{array}{c}
\boldsymbol{a}(\theta_{k})=[1,e^{j2\pi\frac{d}{\lambda}\sin(\theta_{k})},\ldots,e^{j2\pi(N-1)\frac{d}{\lambda}\sin(\theta_{k})}]^{T},\end{array}
\end{equation}
where $\lambda$ denotes the wavelength and $d$ denotes the spacing
between two adjacent antennas. With the \ac{los} consideration, the
channel vector from the BS to target $k$ is denoted by $\boldsymbol{h}_{k}=\alpha_{k}\boldsymbol{a}(\theta_{k})$.
Here, $\alpha_{k}$ denotes the channel amplitude that is given by
$\alpha_{k}=\sqrt{\frac{\Theta_{k}}{D_{k}^{2}}}$, where $\Theta_{k}$
denotes the reference path loss at a distance of $1$ meter, and $D_{k}$
denotes the distance between the \ac{bs} and target $k$ in meters.
The received signal at untrusted target $k\in\mathcal{K}_{E}$ is
denoted as\setlength{\abovedisplayskip}{3pt} \setlength{\belowdisplayskip}{3pt}
\begin{equation}
y_{k}=\boldsymbol{h}_{k}^{H}\boldsymbol{w}_{0}s_{0}+\boldsymbol{h}_{k}^{H}\boldsymbol{s}_{1}+z_{k},
\end{equation}
where $z_{k}\sim\mathcal{CN}(0,\sigma_{k}^{2})$ denotes the AWGN
at the receiver of untrusted target $k$. Accordingly, the received
\ac{sinr} at untrusted target $k$ is\setlength{\abovedisplayskip}{3pt} \setlength{\belowdisplayskip}{3pt}
\begin{equation}
\gamma_{k}(\boldsymbol{w}_{0},\boldsymbol{S})=\frac{\big|\boldsymbol{h}_{k}^{H}\boldsymbol{w}_{0}\big|^{2}}{\boldsymbol{h}_{k}^{H}\boldsymbol{S}\boldsymbol{h}_{k}+\sigma_{k}^{2}}.
\end{equation}

It is assumed that the BS perfectly knows the \ac{csi} $\boldsymbol{g}$
of the \ac{cu} via channel estimation and feedback, and knows the
\ac{csi} $\boldsymbol{h}_{k}$'s of untrusted targets via efficient
 sensing. In this case, the achievable secrecy rate at the \ac{cu}
(in bits-per-second-per-Hertz, bps/Hz) is given by \cite{GopPraJ08}\setlength{\abovedisplayskip}{3pt} \setlength{\belowdisplayskip}{3pt}
\begin{equation}
r(\boldsymbol{w}_{0},\boldsymbol{S})\negthickspace=\negthickspace\underset{k\in\mathcal{K}_{E}}{\min}\negthickspace\big(\log_{2}(1+\gamma_{0}(\boldsymbol{w}_{0},\boldsymbol{S}))-\log_{2}(1+\gamma_{k}(\boldsymbol{w}_{0},\boldsymbol{S}))\big)^{+}.\label{eq:rate}
\end{equation}

Next, we consider the targets sensing, for which the transmit beampattern
gain is used as the performance metric. For any sensing angle $\theta\in[-\frac{\pi}{2},\frac{\pi}{2}]$,
the beampattern gain $P(\theta)$ is defined as the transmit signal
power distribution at $\theta$, i.e.,\setlength{\abovedisplayskip}{3pt} \setlength{\belowdisplayskip}{3pt}
\begin{eqnarray}
P(\theta) & = & \mathbb{E}\left(\big|\boldsymbol{a}^{H}(\theta)(\boldsymbol{s}_{1}+\boldsymbol{w}_{0}s_{0})\big|^{2}\right)\nonumber \\
 & = & \boldsymbol{a}^{H}(\theta)(\boldsymbol{S}+\boldsymbol{w}_{0}\boldsymbol{w}_{0}^{H})\boldsymbol{a}(\theta).
\end{eqnarray}
In particular, with the roughly known targets' locations $\ensuremath{\theta_{k}}$'s,
we define the desired beampattern $\hat{P}(\theta)$ \cite{LiuHuangNirJ20,StoPETLiJ07}
as\setlength{\abovedisplayskip}{3pt} \setlength{\belowdisplayskip}{3pt}
\begin{equation}
\hat{P}(\theta)=\begin{cases}
1 & \exists k\in\mathcal{K},|\theta-\theta_{k}|<\frac{\Delta\theta}{2},\\
0 & \textrm{otherwise},
\end{cases}
\end{equation}
where $\Delta\theta$ denotes the width of beampattern angle. Accordingly,
we use the beampattern matching error as the sensing performance metric,
which measures the difference between the actual transmit beampattern
in angular domain versus the desired beampattern, i.e.,\setlength{\abovedisplayskip}{3pt} \setlength{\belowdisplayskip}{3pt}
\begin{equation}
\begin{array}[b]{l}
B(\{\bar{\theta}_{m}\}_{m=1}^{M},\boldsymbol{w}_{0},\boldsymbol{S},\eta)=\\
\stackrel[m=1]{M}{\sum}\big|\eta\hat{P}(\bar{\theta}_{m})-\boldsymbol{a}^{H}(\bar{\theta}_{m})(\boldsymbol{S}+\boldsymbol{w}_{0}\boldsymbol{w}_{0}^{H})\boldsymbol{a}(\bar{\theta}_{m})\big|^{2},
\end{array}\label{eq:bme}
\end{equation}
where $\{\bar{\theta}_{m}\}_{m=1}^{M}$ denote the $M$ sample angles
over $[-\frac{\pi}{2},\frac{\pi}{2}]$ and $\eta$ is a scaling factor
to be determined. Notice that a larger value of $M$ may lead to more
accurate beampattern matching, but at the cost of increased computation
complexity.

Our objective is to minimize the beampattern matching error in (\ref{eq:bme})
by jointly optimizing the sensing and information beamforming design
($\boldsymbol{w}_{0}$ and $\boldsymbol{S}$), subject to the minimum
secrecy rate constraint of $R_{0}$ and the transmit power constraint
in (\ref{eq:power constraint}). The secrecy rate constrained sensing
beampattern matching problem is thus formulated as\begin{subequations}
\begin{eqnarray}
\textrm{(P1)}: & \!\!\!\underset{\boldsymbol{w}_{0},\boldsymbol{S},\eta}{\min}\!\!\! & \stackrel[m=1]{M}{\sum}\negthickspace\big|\eta\hat{P}(\bar{\theta}_{m})\negmedspace-\negmedspace\boldsymbol{a}^{H}(\bar{\theta}_{m})(\boldsymbol{S}\negmedspace+\negmedspace\boldsymbol{w}_{0}\boldsymbol{w}_{0}^{H})\boldsymbol{a}(\bar{\theta}_{m})\big|^{2}\nonumber \\
 & \textrm{s.t.} & r(\boldsymbol{w}_{0},\boldsymbol{S})\geq R_{0},\label{eq:12}\\
 &  & \textrm{Tr}(\boldsymbol{S})+\boldsymbol{w}_{0}\boldsymbol{w}_{0}^{H}=Q,\boldsymbol{S}\succeq\boldsymbol{0}.\label{eq:13}
\end{eqnarray}
\end{subequations}Notice that problem (P1) is non-convex due to the
non-convex objective function and the non-concave secrecy rate in
(\ref{eq:12}). We will solve problem (P1) optimally in Section III.

Before proceeding, we check the feasibility of problem (P1), which
is equivalent to solving the following secrecy rate maximization problem
(P2).
\begin{eqnarray*}
\textrm{(P2)}: & \!\!\!\underset{\boldsymbol{w}_{0},\boldsymbol{S}}{\max} & r(\boldsymbol{w}_{0},\boldsymbol{S})\\
 & \textrm{s.t.} & \textrm{Tr}(\boldsymbol{S})+\boldsymbol{w}_{0}\boldsymbol{w}_{0}^{H}=Q,\boldsymbol{S}\succeq\boldsymbol{0}.
\end{eqnarray*}
Notice that problem (P2) has been optimally solved in \cite{ChuJ15}.
Let $R^{*}$ denote the optimal objective value or the maximum secrecy
rate achieved in (P2). If $R^{*}\ge R_{0}$, then problem (P1) is
feasible; otherwise, it is infeasible. Therefore, the feasibility
of (P1) has been checked. In the sequel, we focus on the case when
problem (P1) is feasible.\vspace{-4em}

\section{Optimal Joint Beamforming Solution to Problem (P1)\vspace{-3em}}

This section presents the globally optimal joint beamforming solution
to problem (P1) by using the technique of \ac{sdr} together with
a \ac{1d} search. Towards this end, we introduce $\boldsymbol{W}=\boldsymbol{w}_{0}\boldsymbol{w}_{0}^{H}$,
where $\boldsymbol{W}\succeq\boldsymbol{0}$ and $\textrm{rank}(\boldsymbol{W})\le1$.
Accordingly, problem (P1) is equivalently reformulated as \begin{subequations}
\begin{eqnarray}
\textrm{(P3)}: & \!\!\!\underset{\boldsymbol{W},\boldsymbol{S},\eta}{\min}\!\!\! & \negthickspace\negthickspace\stackrel[m=1]{M}{\sum}\negthickspace\big|\eta\hat{P}(\bar{\theta}_{m})\negmedspace-\negmedspace\boldsymbol{a}^{H}(\bar{\theta}_{m})(\boldsymbol{S}\negmedspace+\negmedspace\boldsymbol{W})\boldsymbol{a}(\bar{\theta}_{m})\big|^{2}\nonumber \\
 & \textrm{s.t.} & \negthickspace\negthickspace\log_{2}\big(1\negthickspace+\negthickspace\frac{\boldsymbol{g}^{H}\boldsymbol{W}\boldsymbol{g}}{\boldsymbol{g}^{H}\boldsymbol{S}\boldsymbol{g}+\sigma_{0}^{2}}\big)\nonumber \\
 &  & \negthickspace\negthickspace-\negthickspace\log\big(1\negthickspace+\negthickspace\frac{\boldsymbol{h}_{k}^{H}\boldsymbol{W}\boldsymbol{h}_{k}}{\boldsymbol{h}_{k}^{H}\boldsymbol{S}\boldsymbol{h}_{k}+\sigma_{k}^{2}}\big)\negthickspace\geq\negthickspace R_{0},\negthickspace\forall k\negthickspace\in\negthickspace\mathcal{K}_{E},\label{eq:13a}\\
 &  & \negthickspace\negthickspace\textrm{Tr}(\boldsymbol{S})+\textrm{Tr}(\boldsymbol{W})=Q,\label{eq:1a}\\
 &  & \negthickspace\negthickspace\textrm{rank}(\boldsymbol{W})\leq1,\label{eq:1b}\\
 &  & \negthickspace\negthickspace\boldsymbol{S}\succeq0,\boldsymbol{W}\succeq0.\label{eq:1c}
\end{eqnarray}
\end{subequations}

Next, we further introduce an auxiliary optimization variable $\gamma_{E}>0$,
which denotes the maximum \ac{sinr} at the $K_{E}$ eavesdropping
targets. As a result, problem (P3) is further reformulated as\begin{subequations}
\begin{eqnarray}
(\textrm{P4}):\!\!\! & \negthickspace\underset{\boldsymbol{W},\boldsymbol{S},\eta,\gamma_{E}}{\min}\negthickspace & \negthickspace\negthickspace\negthickspace\stackrel[m=1]{M}{\sum}\negthickspace\big|\eta\hat{P}(\bar{\theta}_{m})\negmedspace-\negmedspace\boldsymbol{a}^{H}(\bar{\theta}_{m})(\boldsymbol{S}\negmedspace+\negmedspace\boldsymbol{W})\boldsymbol{a}(\bar{\theta}_{m})\big|^{2}\nonumber \\
 & \textrm{s.t.} & \negthickspace\negthickspace\negthickspace\negthickspace\boldsymbol{h}_{k}^{H}\boldsymbol{W}\boldsymbol{h}_{k}\negthickspace\leq\negthickspace\gamma_{E}(\boldsymbol{h}_{k}^{H}\boldsymbol{S}\boldsymbol{h}_{k}\negthickspace+\negthickspace\sigma_{k}^{2}),\negthickspace\forall k\in\mathcal{K}_{E},\label{eq:1.1a}\\
 &  & \negthickspace\negthickspace\negthickspace\negthickspace\boldsymbol{g}^{H}\boldsymbol{W}\boldsymbol{g}\geq\beta(\boldsymbol{g}^{H}\boldsymbol{S}\boldsymbol{g}+\sigma_{0}^{2}),\label{eq:1.1b}\\
 &  & \negthickspace\negthickspace\negthickspace\negthickspace\textrm{\textrm{(\ref{eq:1a})}, \textrm{(\ref{eq:1b}), and \textrm{(\ref{eq:1c}),}}}\nonumber 
\end{eqnarray}
\end{subequations} where the two constraints in (\ref{eq:1.1a})
and (\ref{eq:1.1b}) are equivalent to (\ref{eq:13a}), and $\beta=2^{R_{0}}(1+\gamma_{E})-1.$
Under any given $\gamma_{E}>0$, the optimization of $\boldsymbol{W},\boldsymbol{S},\textrm{and }\eta$
in problem (P4) becomes 
\begin{eqnarray}
(\textrm{P4.1}):\!\!\! & \negthickspace\underset{\boldsymbol{W},\boldsymbol{S},\eta}{\min}\negthickspace & \negthickspace\negthickspace\negthickspace\stackrel[m=1]{M}{\sum}\negthickspace\big|\eta\hat{P}(\bar{\theta}_{m})\negmedspace-\negmedspace\boldsymbol{a}^{H}(\bar{\theta}_{m})(\boldsymbol{S}\negmedspace+\negmedspace\boldsymbol{W})\boldsymbol{a}(\bar{\theta}_{m})\big|^{2}\nonumber \\
 & \textrm{s.t.} & \negthickspace\negthickspace\negthickspace\negthickspace\textrm{\textrm{(\ref{eq:1a})}, \textrm{(\ref{eq:1b}), \textrm{(\ref{eq:1c}),} (\ref{eq:1.1a}), and (\ref{eq:1.1b})}}.
\end{eqnarray}
Let $f(\gamma_{E})$ denote the optimal objective value achieved by
problem $(\textrm{P4.1})$ with given $\gamma_{E}$. Accordingly,
we solve problem (P4) by first solving problem (P4.1) under any given
$\gamma_{E}>0$ and then search over $\gamma_{E}$ via a \ac{1d}
search in problem (P4.2).\setlength{\abovedisplayskip}{3pt} \setlength{\belowdisplayskip}{3pt}
\begin{equation}
\begin{array}{ccc}
(\textrm{P4.2}): & \underset{\gamma_{E}>0}{\min} & f(\gamma_{E})\end{array}.
\end{equation}

In the following, we only need to focus on solving problem (P4.1).
By using the \ac{sdr} technique, we relax the rank constraint in
(\ref{eq:1b}) and obtain the \ac{sdr} version of (P4.1) as 
\begin{eqnarray}
(\textrm{SDR4.1}): & \!\!\!\underset{\boldsymbol{W},\boldsymbol{S},\eta}{\min} & \stackrel[m=1]{M}{\sum}\!\!\big|\eta\hat{P}(\bar{\theta}_{m})\!\!-\!\!\boldsymbol{a}^{H}(\bar{\theta}_{m})(\boldsymbol{S}\!\!+\!\!\boldsymbol{W})\boldsymbol{a}(\bar{\theta}_{m})\big|^{2}\nonumber \\
 & \textrm{s.t.} & \textrm{\textrm{(\ref{eq:1a})}, \textrm{(\ref{eq:1c}), (\ref{eq:1.1a}), and (\ref{eq:1.1b}).}}
\end{eqnarray}
It is observed that problem $(\textrm{SDR4.1)}$ is a convex \ac{qsdp}
problem that can be optimally solved by convex solvers, such as CVX
\cite{cvx}. Let $\tilde{\boldsymbol{W}}^{*}$, $\tilde{\boldsymbol{S}}^{*}$,
$\textrm{ and }\tilde{\eta}^{*}$ denote the obtained optimal solution
to problem (SDR4.1), where $\textrm{rank}(\tilde{\boldsymbol{W}}^{*})\leq1$
may not hold in general. In this case, we need additional steps to
construct rank-one solutions to problem (P4.1). In the following proposition,
we show that the SDR is tight and accordingly construct an optimal
rank-one solution to problem (SDR4.1) and thus (P4.1).
\begin{prop}
\textup{Based on the obtained optimal solution $\tilde{\boldsymbol{W}}^{*}$,
$\tilde{\boldsymbol{S}}^{*}$, $\textrm{ and }\tilde{\eta}^{*}$ to
problem (SDR4.1) with $\textrm{rank}(\tilde{\boldsymbol{W}}^{*})>1$
in general, we can always construct an equivalent solution of $\boldsymbol{W}^{*}$,
$\boldsymbol{S}^{*}$, and $\eta^{*}$ in the following, which is
optimal for problems (SDR4.1) and (P4.1) with $\textrm{rank}(\boldsymbol{W}^{*})=1$.\begin{subequations}
\begin{eqnarray}
 &  & \boldsymbol{w}_{0}^{*}=(\boldsymbol{g}^{H}\tilde{\boldsymbol{W}}^{*}\boldsymbol{g})^{-\frac{1}{2}}\tilde{\boldsymbol{W}}^{*}\boldsymbol{g},\label{eq:c1}\\
 &  & \boldsymbol{W}^{*}=\boldsymbol{w}_{0}^{*}\boldsymbol{w}_{0}^{*H}=\frac{\tilde{\boldsymbol{W}}^{*}\boldsymbol{g}\boldsymbol{g}^{H}\tilde{\boldsymbol{W}}^{*}}{\boldsymbol{g}^{H}\tilde{\boldsymbol{W}}^{*}\boldsymbol{g}},\label{eq:c2}\\
 &  & \boldsymbol{S}^{*}=\tilde{\boldsymbol{W}}^{*}+\tilde{\boldsymbol{S}}^{*}-\boldsymbol{W}^{*},\label{eq:c3}\\
 &  & \eta^{*}=\tilde{\eta}^{*}.\label{eq:c4}
\end{eqnarray}
\end{subequations}}
\end{prop}
\begin{proof}
See Appendix A.
\end{proof}
Suppose that the optimal solution of $\gamma_{E}$ in problem (P4.2)
is given by $\gamma_{E}^{*}$. In this case, by solving problem (P4.1)/(SDR4.1)
with given $\gamma_{E}^{*}$, the correspondingly constructed optimal
solution of $\boldsymbol{W}^{*}$, $\boldsymbol{S}^{*}$, and $\eta^{*}$
in Proposition 1 becomes the optimal solution to problems (P4) and
(P3). Accordingly, the corresponding $\boldsymbol{w}_{0}^{*}$ from
(\ref{eq:c1}), $\boldsymbol{S}^{*}$, and $\eta^{*}$ are the obtained
optimal solution to problem (P1). 

\section{Suboptimal Beamforming Solutions to Problem (P1)}

The preceding section proposed the globally optimal solution to problem
(P1), which, however, suffers from high computational complexity due
to the \ac{1d} search. To overcome this issue, this section presents
two low-complexity designs to obtain high-quality sub-optimal solutions
based on \ac{zf} and separate beamforming, respectively. 

\subsection{ZF-Based Beamforming Design}

First, we consider the \ac{zf}-based beamforming design, in which
the information beamforming vector $\boldsymbol{w}_{0}$ is enforced
to lie in the null space of all eavesdroppers' channel vectors, i.e.$,\boldsymbol{h}_{k}^{H}\boldsymbol{w}_{0}=0,\forall k\in\mathcal{K}_{E}$,
such that the eavesdroppers are not able to receive any confidential
information from the BS. Notice that the \ac{zf}-based beamforming
design only works when the number of transmit antennas $N$ at the
BS is greater than the number of eavesdroppers $K_{E}$, i.e., $N>K_{E}$. 

Let $\boldsymbol{H}=[\boldsymbol{h}_{1},\boldsymbol{h}_{2},\ldots,\boldsymbol{h}_{K_{E}}]^{H}$
denote the channel matrix from the BS to all the eavesdroppers, of
which the singular value decomposition (SVD) is
\begin{equation}
\boldsymbol{H}=\boldsymbol{\boldsymbol{U}\Lambda}\boldsymbol{V}^{H}=\boldsymbol{\boldsymbol{U}\Lambda}[\boldsymbol{V}_{1}\boldsymbol{V}_{2}]^{H},
\end{equation}
where $\boldsymbol{\Lambda}\in C^{K_{E}\times N}$ has non-zero diagonal
elements that correspond the singular values of $\boldsymbol{H}$,
$\boldsymbol{U}\in\mathbb{C}^{K_{E}\times K_{E}}$ and $\boldsymbol{V}\in\mathbb{C}^{N\times N}$
are both unitary matrices, and $\boldsymbol{V}_{1}\in\mathbb{C}^{N\times K_{E}}$
and $\boldsymbol{V}_{2}\in\mathbb{C}^{N\times(N-K_{E})}$ consist
of the first $K_{E}$ and the last $N-K_{E}$ right singular vectors
of $\boldsymbol{H}$, respectively. In order to ensure $\boldsymbol{h}_{k}^{H}\boldsymbol{w}_{0}=0,\forall k\in\mathcal{K}_{E}$,
we set 
\begin{equation}
\setlength{\abovedisplayskip}{3pt}\setlength{\belowdisplayskip}{3pt}\boldsymbol{w}_{0}=\boldsymbol{V}_{2}\bar{\boldsymbol{w}}_{0}
\end{equation}
without loss of generality, where $\bar{\boldsymbol{w}}_{0}\in\mathbb{C}^{(N-K_{E})\times1}$
to be optimized. In this case, the secrecy rate becomes 
\begin{equation}
\bar{r}(\bar{\boldsymbol{w}}_{0},\boldsymbol{S})\negthickspace=\log_{2}\big(1+\frac{\boldsymbol{g}^{H}\boldsymbol{V}_{2}\bar{\boldsymbol{w}}_{0}\bar{\boldsymbol{w}}_{0}^{H}\boldsymbol{V}_{2}^{H}\boldsymbol{g}}{\sigma_{0}^{2}+\boldsymbol{g}^{H}\boldsymbol{S}\boldsymbol{g}}\big).
\end{equation}
To maximize the secrecy rate, the beamformer $\bar{\boldsymbol{w}}_{0}$
can be set as $\bar{\boldsymbol{w}}_{0}=\sqrt{Q_{0}}\frac{\boldsymbol{V}_{2}^{H}\boldsymbol{g}}{\|\boldsymbol{V}_{2}^{H}\boldsymbol{g}\|}$,
and accordingly we have\setlength{\abovedisplayskip}{3pt}\setlength{\belowdisplayskip}{3pt}

\begin{equation}
\boldsymbol{w}_{0}=\sqrt{Q_{0}}\boldsymbol{V}_{2}\frac{\boldsymbol{V}_{2}^{H}\boldsymbol{g}}{\|\boldsymbol{V}_{2}^{H}\boldsymbol{g}\|},\label{eq:30}
\end{equation}
where $Q_{0}\ge0$ is the optimizable transmit power for information
signals. Hence, by substituting the ZF beamforming design of $\boldsymbol{w}_{0}$
in (\ref{eq:30}) in problem (P1), the secrecy rate constrained sensing
beampattern matching problem is recast into\setlength{\abovedisplayskip}{3pt}\setlength{\belowdisplayskip}{3pt}$ $
\begin{eqnarray}
\negthickspace\negthickspace\textrm{(P5)}: & \negthickspace\negthickspace\underset{Q_{0},\boldsymbol{S},\eta}{\min} & \negthickspace\negthickspace\stackrel[m=1]{M}{\sum}\negthickspace\big|\eta\hat{P}(\bar{\theta}_{m})\negthickspace-\negthickspace\boldsymbol{a}^{H}(\bar{\theta}_{m})(\boldsymbol{S}\negthickspace+\negthickspace Q_{0}\tilde{\boldsymbol{w}}_{0}\tilde{\boldsymbol{w}}_{0}^{H})\boldsymbol{a}(\bar{\theta}_{m})\big|^{2}\nonumber \\
 & \textrm{s.t.} & \negthickspace\negthickspace Q_{0}\|\boldsymbol{V}_{2}^{H}\boldsymbol{g}\|^{2}\geq(\sigma_{0}^{2}+\boldsymbol{g}^{H}\boldsymbol{S}\boldsymbol{g})(2^{R_{0}}-1),\nonumber \\
 &  & \negthickspace\negthickspace\textrm{Tr}(\boldsymbol{S})+Q_{0}=Q,\boldsymbol{S}\succeq0,
\end{eqnarray}
where $\tilde{\boldsymbol{w}}_{0}=\boldsymbol{V}_{2}\frac{\boldsymbol{V}_{2}^{H}\boldsymbol{g}}{\|\boldsymbol{V}_{2}^{H}\boldsymbol{g}\|}$.
Problem $\textrm{(P5)}$ is also a convex \ac{qsdp} problem that
can be optimally solved by CVX. Let $\bar{Q}_{0}^{*}$, $\bar{\boldsymbol{S}}^{*},$
and $\bar{\eta}^{*}$ denote the obtained optimal solution to problem
(P5). Accordingly, based on (\ref{eq:30}) the ZF-based beamforming
solution is obtained as\setlength{\abovedisplayskip}{3pt}\setlength{\belowdisplayskip}{3pt}
\begin{eqnarray}
\boldsymbol{w}_{0}=\sqrt{\bar{Q}_{0}^{*}}\boldsymbol{V}_{2}\frac{\boldsymbol{V}_{2}^{H}\boldsymbol{g}}{\|\boldsymbol{V}_{2}^{H}\boldsymbol{g}\|},\boldsymbol{S}=\bar{\boldsymbol{S}}^{*},\eta=\bar{\eta}^{*}.
\end{eqnarray}

\subsection{Separate Beamforming Design}

In this subsection, we propose another sub-optimal beamforming design,
in which the information beamforming vector $\boldsymbol{w}_{0}$
and the sensing covariance matrix $\boldsymbol{S}$ are designed separately.
In this approach, we first design the information beamforming vector
$\boldsymbol{w}_{0}$ to achieve the secrecy rate with the minimum
required power, and then use the remaining power for sensing signals
$\boldsymbol{S}$. 

First, we design the information beamforming $\boldsymbol{w}_{0}$.
Without sensing signals or \ac{an}, the secrecy rate is given as\setlength{\abovedisplayskip}{3pt} \setlength{\belowdisplayskip}{3pt}
\begin{equation}
\tilde{r}(\boldsymbol{w}_{0})=\underset{k\in\mathcal{K}_{E}}{\min}\negthickspace\big(\log_{2}(1+\bar{\gamma}_{0}(\boldsymbol{w}_{0}))-\log_{2}(1+\bar{\gamma}_{k}(\boldsymbol{w}_{0}))\big)^{+},
\end{equation}
where $\bar{\gamma}_{0}(\boldsymbol{w}_{0})=\frac{\big|\boldsymbol{g}^{H}\boldsymbol{w}_{0}\big|^{2}}{\sigma_{0}^{2}}$
denotes the \ac{snr} at the CU and $\bar{\gamma}_{k}(\boldsymbol{w}_{0})=\frac{\big|\boldsymbol{h}_{k}^{H}\boldsymbol{w}_{0}\big|^{2}}{\sigma_{k}^{2}}$
denotes the \ac{snr} at eavesdropping target $k\in\mathcal{K}_{E}$.
In this case, the information beamforming vector $\boldsymbol{w}_{0}$
is designed by minimizing the transmit power $\|\boldsymbol{w}_{0}\|^{2}$
while ensuring the minimum secrecy rate of $R_{0}$, i.e., 
\begin{eqnarray}
\textrm{(P6)}: & \negthickspace\underset{\boldsymbol{w}_{0}}{\min}\negthickspace & \|\boldsymbol{w}_{0}\|^{2}\nonumber \\
 & \textrm{s.t.} & \tilde{r}(\boldsymbol{w}_{0})\geq R_{0},\nonumber \\
 &  & \|\boldsymbol{w}_{0}\|^{2}\leq Q.
\end{eqnarray}
It is observed that problem (P6) is a typical multiple-input single-output
(MISO) secrecy communication problem and has been optimally solved
in \cite{ChuJ15}. Let $\boldsymbol{w}_{0}^{\star}$ denote the obtained
optimal solution beamforming to problem (P6).

Next, we design the sensing covariance matrix $\boldsymbol{S}$ without
interfering with the confidential information signals reception at
the \ac{cu}, i.e., $\boldsymbol{g}^{H}\boldsymbol{S}\boldsymbol{g}=0$.
Without loss of generality, we define $\boldsymbol{Q}_{2}=\boldsymbol{I}-\boldsymbol{g}\boldsymbol{g}^{H}/\|\boldsymbol{g}\|^{2}$,
and accordingly set 
\begin{equation}
\boldsymbol{S}=\boldsymbol{Q}_{2}\bar{\boldsymbol{S}}\boldsymbol{Q}_{2}^{H},\label{eq:35}
\end{equation}
where $\bar{\boldsymbol{S}}\in\mathbb{C}^{N\times N}$ and $\bar{\boldsymbol{S}}\succeq0$.
By substituting $\boldsymbol{w}_{0}^{\star}$ and $\boldsymbol{S}$
in (\ref{eq:35}) to problem (P1), we have the following beampattern
matching error minimization problem.
\begin{eqnarray}
\textrm{(P7)}\!:\negthickspace\negthickspace & \negthickspace\negthickspace\underset{\bar{\boldsymbol{S}},\eta}{\min}\negthickspace & \negthickspace\negthickspace\stackrel[m=1]{M}{\sum}\!\!\!\big|\eta\hat{P}(\bar{\theta}_{m})\!\!-\!\!\boldsymbol{a}^{H}(\bar{\theta}_{m})(\boldsymbol{Q}_{2}\bar{\boldsymbol{S}}\boldsymbol{Q}_{2}^{H}\!\thinspace\!+\!\!\boldsymbol{w}_{0}^{\star}\boldsymbol{w}_{0}^{\star H})\boldsymbol{a}(\bar{\theta}_{m})\big|^{2}\nonumber \\
 & \textrm{s.t.} & \negthickspace\negthickspace\bar{\boldsymbol{S}}\succeq0,\nonumber \\
 &  & \negthickspace\negthickspace\textrm{Tr}(\boldsymbol{Q}_{2}\bar{\boldsymbol{S}}\boldsymbol{Q}_{2}^{H})+\|\boldsymbol{w}_{0}^{\star}\|^{2}=Q.
\end{eqnarray}
Problem (P7) is also a \ac{qsdp} problem that can be optimally solved
by CVX. Let $\bar{\boldsymbol{S}}^{\star}$ denote the optimal solution
to problem (P7). Then $\boldsymbol{S}^{\star}=\boldsymbol{Q}_{2}\bar{\boldsymbol{S}}^{\star}\boldsymbol{Q}_{2}^{H}$
and $\boldsymbol{w}_{0}^{\star}$ become the obtained separate beamforming
designs. 

It is worth remarking that the two sub-optimal beamforming designs
in Sections IV-A and IV-B do not require the computational-heavy \ac{1d}
search in the optimal solution, thus significantly reducing the computational
complexity. 

\section{Numerical Results}

This section provides numerical results to validate the performance
of our proposed joint information and sensing beamforming designs
for secrecy \ac{isac}. In the simulation, the BS is deployed with
$N=8$ antenna elements, and there are $K=8$ targets located at angles
$-10\lyxmathsym{\textdegree},10\lyxmathsym{\textdegree},-30\lyxmathsym{\textdegree},30\lyxmathsym{\textdegree},80\lyxmathsym{\textdegree},-80\lyxmathsym{\textdegree},-50\lyxmathsym{\textdegree},\textrm{and }50\lyxmathsym{\textdegree}$.
The two targets located at $-30\lyxmathsym{\textdegree}\textrm{ and }30\lyxmathsym{\textdegree}$
are assumed to be untrusted eavesdroppers. We set the noise powers
at the \ac{cu} and all the eavesdroppers are identical to be $-60\textrm{dBm},$
i.e., $\sigma_{0}^{2}=\sigma_{k}^{2}=-60\textrm{dBm},\forall k\in\mathcal{K}_{E}$,
and the path loss from the BS to the CU and the eavesdroppers to be
$-70\textrm{dB}$. Furthermore, the transmit power budget at the BS
is set to be $Q=20\textrm{dBm}$. We choose $M=201$ angles for $\{\bar{\theta}_{m}\}$,
which are uniformly sampled in $[-\frac{\pi}{2},\frac{\pi}{2}]$.
We also set the beampattern width as $\Delta\theta=5\lyxmathsym{\textdegree}$.
The normalized spacing between two adjacent antennas is set as $\frac{d}{\lambda}=0.5.$
The LoS model is considered for the channel from the BS to the CU.
\vspace{-0.3cm} 
\begin{figure}[tbh]
\includegraphics[scale=0.5]{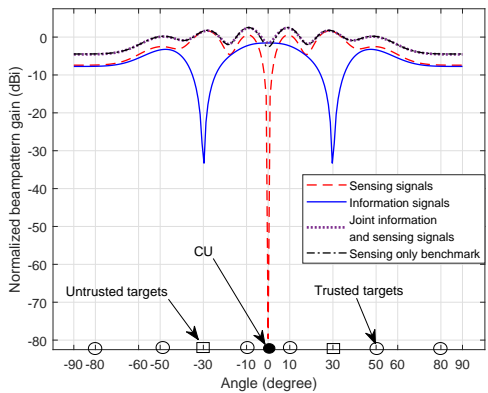}\centering

\caption{\label{fig:5}The normalized beampattern gains achieved by the proposed
optimal beamforming design, where the CU is located at $\theta_{0}=0\lyxmathsym{\protect\textdegree}$,
and the secrecy rate threshold is $R_{0}=3\textrm{bps/Hz}$. Here,
the beampattern gain by the information signals is defined as $\boldsymbol{a}^{H}(\theta)\boldsymbol{w}_{0}^{*}\boldsymbol{w}_{0}^{*H}\boldsymbol{a}(\theta)$,
and that by the sensing signals is defined as $\boldsymbol{a}^{H}(\theta)\boldsymbol{S}^{*}\boldsymbol{a}(\theta)$.}
\end{figure}
\begin{figure}[tbh]
\includegraphics[scale=0.5]{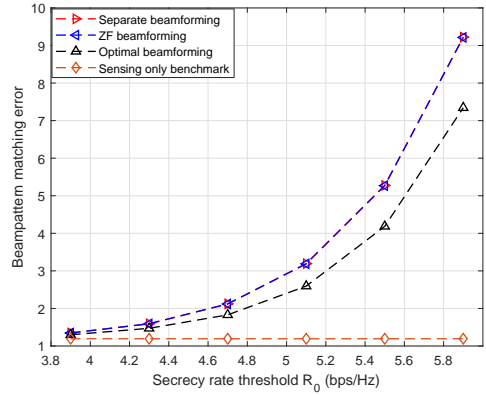}\centering\caption{\label{fig:3}The beampattern matching error versus the secrecy rate
threshold $R_{0}$ under different designs.}
\end{figure}
\begin{figure}[tbh]
\includegraphics[scale=0.45]{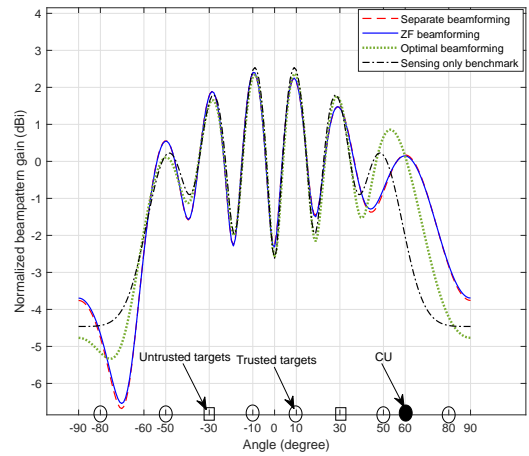}\centering

\caption{\label{fig:4}The normalized beampattern gains under different designs,
where the CU is located at $\theta_{0}=60\lyxmathsym{\protect\textdegree}$,
and the secrecy rate threshold is $R_{0}=3\textrm{bps/Hz}$.}
\end{figure}

Fig. \ref{fig:5} shows the normalized beampattern gains achieved
by the proposed optimal beamforming design solution, where the \ac{cu}
is located at $\theta_{0}=0{^\circ}$ and the secrecy rate threshold
is set as $R_{0}=3\textrm{bps/Hz}$. For illustration, we also show
the beampattern gain solely achieved by information signals (i.e.,
$\boldsymbol{a}^{H}(\theta)\boldsymbol{w}_{0}^{*}\boldsymbol{w}_{0}^{*H}\boldsymbol{a}(\theta)$)
and that solely achieved by sensing signals (i.e., $\boldsymbol{a}^{H}(\theta)\boldsymbol{S}^{*}\boldsymbol{a}(\theta)$).
For comparison, we also consider the sensing only benchmark, which
corresponds to the design in problem (P1) with $R_{0}=0$ and $\boldsymbol{w}_{0}=\boldsymbol{0}$.
It is observed that the beampattern gain achieved by the information
signals is around $30\textrm{dB}$ lower at angles of $-30\lyxmathsym{\textdegree}$
and $30\lyxmathsym{\textdegree}$. This is for the purpose of preventing
the information leakage as the eavesdropping targets are located at
these two angles. It is also observed that the beampattern gain achieved
by the sensing signals is around $70\textrm{dB}$ lower at $0\lyxmathsym{\textdegree}$.
This is in order to not interfere with the \ac{cu} at that angle.
By combining them, the joint information and sensing signals are observed
to achieve similar beampattern gains as those by the sensing only
benchmark. This shows the effectiveness of the proposed optimal solution. 

Fig. 2 shows the beampattern matching error versus the secrecy rate
threshold $R_{0}$ under different designs, where the \ac{cu} is
located at $\theta_{0}=60\lyxmathsym{\textdegree}$. It is observed
that as $R_{0}$ becomes large, the achieved beampattern matching
errors for the three proposed designs increases, and the performance
gap between the optimal solution versus the sub-optimal ones becomes
more significant. It is also observed that as $R_{0}$ becomes small,
such performance gap becomes small, and all the three designs approach
the beampattern gain error lower bound by the sensing only benchmark.

Fig. \ref{fig:4} shows the normalized beampattern gains under different
designs, with $R_{0}=3.5\textrm{bps/Hz}$ and $\theta_{0}=60\lyxmathsym{\textdegree}$.
It is observed that the beampattern gains by the two sub-optimal designs
well match those by the optimal solution and sensing only benchmark
in most sensing angles. This validates the effectiveness of the proposed
designs.\addtolength{\topmargin}{0.1in}

\section{Conclusion}

This paper studied the joint transmit information and sensing beamforming
design for secrecy ISAC system with one CU and multiple untrusted
and trusted targets. Our objective was to minimize the sensing beampattern
matching error, while ensuring the secrecy rate requirement. We proposed
the globally optimal solution to the highly non-convex problem, via
the \ac{sdr} and \ac{1d} search, in which the tightness of \ac{sdr}
is proved. We also proposed two sub-optimal solutions based on ZF
and separate beamforming, respectively. Numerical results showed that
at the proposed solutions, the information beams are designed towards
the CU and the trusted targets, while the sensing beams are designed
towards both untrusted and trusted targets, thus increasing the sensing
performance while ensuring the CU's secrecy rate. How to extend the
proposed designs to setups with imperfect CSI, multiple CUs, and/or
multuple BSs are interesting research directions worth future pursuing.

\appendices{}

\section{Proof of Proposition 1}

First, it is observed from (\ref{eq:c3}) that $\boldsymbol{S}^{*}+\boldsymbol{W}^{*}=\tilde{\boldsymbol{W}}^{*}+\tilde{\boldsymbol{S}}^{*}$.
As a result, $\boldsymbol{W}^{*}$, $\boldsymbol{S}^{*}$, and $\eta^{*}$
achieve the same objective value for problem (SDR4.1) as that achieved
by $\tilde{\boldsymbol{W}}^{*}$, $\tilde{\boldsymbol{S}}^{*}$, and
$\tilde{\eta}^{*}$, and $\boldsymbol{W}^{*}$, $\boldsymbol{S}^{*}$,
and $\eta^{*}$ satisfy the constraint in (\ref{eq:1a}). 

Next, based on $\tilde{\boldsymbol{W}}^{*}\succeq\boldsymbol{0}$,
we have $\tilde{\boldsymbol{W}}^{*}=\bar{\boldsymbol{W}}\bar{\boldsymbol{W}}^{H}$.
Based on this, for any $\boldsymbol{v}\in\mathbb{C}^{N\times1}$,
it follows that 
\begin{equation}
\begin{array}[b]{l}
\boldsymbol{v}^{H}(\tilde{\boldsymbol{W}}^{*}-\boldsymbol{W}^{*})\boldsymbol{v}=\boldsymbol{v}^{H}\tilde{\boldsymbol{W}^{*}}\boldsymbol{v}-\boldsymbol{v}^{H}\frac{\tilde{\boldsymbol{W}^{*}}\boldsymbol{g}\boldsymbol{g}^{H}\tilde{\boldsymbol{W}^{*}}}{\boldsymbol{g}^{H}\tilde{\boldsymbol{W}^{*}}\boldsymbol{g}}\boldsymbol{v}\\
=\frac{1}{\boldsymbol{g}^{H}\tilde{\boldsymbol{W}^{*}}\boldsymbol{g}}(\boldsymbol{v}^{H}\tilde{\boldsymbol{W}^{*}}\boldsymbol{v}\boldsymbol{g}^{H}\tilde{\boldsymbol{W}^{*}}\boldsymbol{g}-\boldsymbol{v}^{H}\tilde{\boldsymbol{W}^{*}}\boldsymbol{g}\boldsymbol{g}^{H}\tilde{\boldsymbol{W}^{*}}\boldsymbol{v})\\
=\frac{1}{\boldsymbol{g}^{H}\tilde{\boldsymbol{W}^{*}}\boldsymbol{g}}(\|\boldsymbol{a}\|^{2}\|\boldsymbol{b}\|^{2}-\big|\boldsymbol{a}^{H}\boldsymbol{b}\big|^{2})\overset{\textrm{(a)}}{\geq}0,
\end{array}
\end{equation}
where $\boldsymbol{a}=\bar{\boldsymbol{W}}^{H}\boldsymbol{v}\in\mathbb{C}^{N\times1},\boldsymbol{b}=\bar{\boldsymbol{W}}^{H}\boldsymbol{g}\in\mathbb{C}^{N\times1}$,
and inequality (a) holds because of the Cauchy-Schwartz inequality.
Accordingly, we have $\tilde{\boldsymbol{W}}^{*}-\boldsymbol{W}^{*}\succeq\boldsymbol{0}$.
By using this together with (\ref{eq:c3}), we have $\boldsymbol{S}^{*}\succeq\tilde{\boldsymbol{S}}^{*}\succeq\boldsymbol{0}$.
Therefore, constraint (\ref{eq:1c}) holds for $\boldsymbol{S}^{*}$
and $\boldsymbol{W}^{*}$. 

Furthermore, we prove that $\boldsymbol{W}^{*}$, $\boldsymbol{S}^{*}$,
and $\eta^{*}$ satisfy the constraints in (\ref{eq:1.1a}) and (\ref{eq:1.1b}).
On one hand, it is clear from (\ref{eq:c2}) and (\ref{eq:c3}) that
$\boldsymbol{g}^{H}\tilde{\boldsymbol{W}^{*}}\boldsymbol{g}=\boldsymbol{g}^{H}\boldsymbol{W}^{*}\boldsymbol{g}$
and $\boldsymbol{g}^{H}\tilde{\boldsymbol{S}}^{*}\boldsymbol{g}=\boldsymbol{g}^{H}\boldsymbol{S}{}^{*}\boldsymbol{g}$,
and therefore, constraint (\ref{eq:1.1b}) holds. On the other hand,
based on $\tilde{\boldsymbol{W}^{*}}-\boldsymbol{W}^{*}\succeq0$
and $\boldsymbol{S}^{*}\succeq\tilde{\boldsymbol{S}}^{*}\succeq\boldsymbol{0}$,
it follows that 

\begin{equation}
\begin{aligned}[b] & \boldsymbol{h}_{k}^{H}\boldsymbol{W}^{*}\boldsymbol{h}_{k}\leq\boldsymbol{h}_{k}^{H}\tilde{\boldsymbol{W}}^{*}\boldsymbol{h}_{k}\\
 & \leq\gamma_{E}(\boldsymbol{h}_{k}^{H}\tilde{\boldsymbol{S}}^{*}\boldsymbol{h}_{k}+\sigma_{k}^{2})\\
 & \leq\gamma_{E}(\boldsymbol{h}_{k}^{H}\boldsymbol{S}^{*}\boldsymbol{h}_{k}+\sigma_{k}^{2}),\forall k\in\mathcal{K}_{E}.
\end{aligned}
\label{eq:24}
\end{equation}
As a result, constraint (\ref{eq:1.1a}) is satisfied. 

By combining the results above, it is proved that $\boldsymbol{W}^{*}$,
$\boldsymbol{S}^{*}$, and $\eta^{*}$ are optimal for problem (SDR4.1).
Notice that $\textrm{rank}(\boldsymbol{W}^{*})=1$ with $\boldsymbol{W}^{*}=\boldsymbol{w}_{0}^{*}\boldsymbol{w}_{0}^{*H}$.
Therefore, $\boldsymbol{W}^{*}$, $\boldsymbol{S}^{*}$, and $\eta^{*}$
are also optimal for problem (P4.1). This thus completes the proof.

\appendices{}

\footnotesize\bibliographystyle{IEEEtran}
\bibliography{IEEEabrv,IEEEexample,my_ref}

\end{document}